\documentclass[11pt]{article}
\usepackage[letterpaper]{geometry}
\usepackage{graphicx} 
\usepackage{hyperref}
\usepackage{xcolor}
\usepackage{mathrsfs}
\usepackage{amsthm}
\usepackage{amsmath}
\usepackage{amssymb}
\usepackage{natbib}
\usepackage{mathtools}
\usepackage{geometry}
\usepackage{listings}
\usepackage{xparse}
\usepackage{setspace}
\usepackage{algorithmicx}
\usepackage{algpseudocode}
\usepackage{booktabs}
\usepackage{algorithm}
\usepackage{authblk}
\usepackage{mathabx}
\usepackage[english]{babel}
\newtheorem{theorem}{Theorem}
\theoremstyle{remark}

\newtheorem{corollary}{Corollary}

\newtheorem{lemma}{Lemma}

\title{ Coordinate Descent Algorithm for Least Absolute Deviations Regression}

\author{
    Zehaan Naik\thanks{Department of Mathematics and Statistics, 
    Indian Institute of Technology Kanpur, India. 
    
    Email: zehaan.naik@gmail.com}
    \and
    Debasis Kundu\thanks{Department of Mathematics and Statistics, 
    Indian Institute of Technology Kanpur, India. 
    
    Email: kundu@iitk.ac.in}
}
\date{\today} 

\begin{document} 
\onehalfspacing

\maketitle 

\begin{abstract}
Least Absolute Deviations (LAD) regression provides a robust alternative to ordinary least squares by minimizing the sum of absolute residuals. However, its widespread use has been limited by the computational cost of existing solvers, particularly simplex-based methods in high-dimensional settings. We propose a coordinate descent algorithm for LAD regression that avoids matrix inversion, naturally accommodates the non-differentiability of the objective function, and remains well-defined even when the number of predictors exceeds the number of observations. The key observation is that each coordinate update reduces to a one-dimensional minimization admitting a closed-form solution given by a median or weighted median. The resulting algorithm has per-iteration complexity \(O(p\,n \log n)\) and is provably convergent due to the convexity of the LAD objective and the exactness of each coordinate update. Experiments on synthetic and real datasets show that the method matches the accuracy of linear-programming-based LAD solvers while offering improved scalability and stability in high-dimensional regimes, including cases where \(p \ge n\). The method is easy to implement, requires no specialized optimization software, and provides a practical tool for robust linear models.
\end{abstract}

\section{Introduction}
Linear regression is a central tool in statistical modeling, used to relate a response variable to a set of covariates through a linear predictor. The most widely used estimator in this setting is Ordinary Least Squares (OLS), which minimizes the sum of squared residuals and admits a closed-form solution under mild conditions. Under the classical linear model assumptions (most notably homoscedastic errors with finite second moments), OLS enjoys optimality properties, as guaranteed by the Gauss-Markov theorem, yielding the minimum-variance linear unbiased estimator. However, because squared residuals place increasing weight on large deviations, the OLS estimator can be highly sensitive to outliers or heavy-tailed noise. This sensitivity motivates the study of alternative estimators that retain the interpretability of linear models while providing improved robustness.

A natural robust alternative to OLS is the Least Absolute Deviations (LAD) estimator, which minimizes the sum of absolute residuals and corresponds to median regression. By penalizing deviations linearly rather than quadratically, LAD reduces the influence of extreme observations and remains well-behaved under heavy-tailed error distributions. Despite these robustness properties, LAD regression has historically been approached primarily through linear programming formulations. Classical methods solve the resulting optimization problem using simplex or interior-point algorithms, as in the foundational work of \cite{charnes1955}, \cite{BarrodaleRoberts1973}, and \cite{koenker1978}. While these approaches guarantee global optimality, they rely on specialized linear programming solvers and can become computationally expensive as the problem dimension grows \citep{koenker2005, hastie2009}. Moreover, in high-dimensional or rank-deficient settings, particularly when the number of predictors is comparable to or exceeds the number of observations, LP formulations may suffer from numerical ill-conditioning or reduced reliability, limiting their practical scalability \citep{portnoy1997, koenker2005}.

To address these limitations, we propose a simple, scalable coordinate descent algorithm for LAD regression. We consider the standard linear regression model
\[
y_i = x_i^\top \beta + \varepsilon_i, \qquad i = 1, \dots, n,
\]
where $y_i \in \mathbb{R}$ is the response, $x_i \in \mathbb{R}^p$ is a vector of covariates (including an intercept when appropriate), and $\beta \in \mathbb{R}^p$ is an unknown coefficient vector. Let $X \in \mathbb{R}^{n\times p}$ denote the design matrix. This notation is used across the article. No parametric assumptions are imposed on the distribution of the error terms in implementing the proposed algorithm except $E(\varepsilon_i) = 0$. However, additional assumptions on the error distribution may be required when deriving asymptotic statistical properties of the estimator. The Least Absolute Deviations (LAD) estimator is defined as any solution to
\[
\widehat{\beta}_{\mathrm{LAD}}
=
\arg\min_{\beta \in \mathbb{R}^p}
\sum_{i=1}^n \bigl| y_i - x_i^\top \beta \bigr|,
\]
which corresponds to estimating the conditional median of $y_i$ given $x_i$. Although this objective is convex, it is generally nondifferentiable and may admit multiple minimizers, particularly in high-dimensional or rank-deficient settings.

To compute LAD estimators efficiently, we exploit the structure of the objective. Fixing all coordinates except $\beta_j$ yields the univariate subproblem
\[
\beta_j
\;\in\;
\arg\min_{t \in \mathbb{R}}
\sum_{i=1}^n \bigl| r_i^{(j)} - x_{ij} t \bigr|,
\qquad
r_i^{(j)} = y_i - \sum_{k \neq j} x_{ik}\beta_k,
\]
whose solution is given by a median or weighted median. This leads naturally to a coordinate descent scheme \citep{Tseng2001, Friedman2010} with exact updates (see Algorithm~\ref{alg:LAD_naive}). By maintaining an incremental residual, the resulting algorithm avoids matrix inversion and reduces the cost of a full coordinate sweep to \(O(p\,n \log n)\). Owing to the convexity of the LAD objective and the exactness of each coordinate update, the method is provably convergent to a global minimizer. Empirically, it remains stable in high-dimensional and underdetermined regimes, including cases where \(p \ge n\), and achieves predictive performance comparable to LP-based LAD solvers when those methods are applicable. We also provide proof of convergence of this method.

We assess the proposed method through synthetic experiments and real-world benchmarks commonly used in the robust and quantile regression literature, including the Boston Housing, Air Quality, and Concrete Compressive Strength datasets \citep{koenker1978, koenker2005, huber2009, yeh1998}. Across these settings, the coordinate descent algorithm attains predictive accuracy comparable to LP-based LAD solvers when applicable, while remaining stable in high-dimensional regimes where LP formulations become unreliable. In experiments with \(p \ge n\), ridge-initialized LAD-CD consistently improves upon unrefined initial estimates, with reductions in prediction error of up to 30-40\%. These results complement the theoretical convergence guarantees and demonstrate the practical effectiveness of the proposed approach.

The remainder of the paper is organized as follows. Section~\ref{sec:background} reviews the LAD regression problem and the univariate optimization results that underpin coordinate-wise updates. Section~\ref{sec:methodology} introduces the proposed coordinate descent algorithm, establishes its convergence properties, and presents computational optimizations and warm-start strategies. Section~\ref{sec:experiments} reports empirical results on synthetic and real-world datasets, including high-dimensional and outlier-contaminated settings. Section~\ref{sec:conclusion} concludes with a discussion of practical implications and directions for future work.

\section{Background}
\label{sec:background}

This section reviews the Least Absolute Deviations (LAD) regression problem and summarizes classical solution strategies that motivate our proposed approach. We focus on two fundamental univariate optimization problems whose solutions form the basis of the coordinate descent algorithm developed in subsequent sections.

\subsection{Minimizing Deviations from a Constant}

The first and most basic subproblem is to find a single constant $a \in \mathbb{R}$ that best represents a set of observations $\{y_1, \dots, y_n\}$ in the $L_1$ sense. This amounts to minimizing
\begin{equation}
\mathcal{L}(a) = \sum_{i=1}^{n} |y_i - a|.
\label{Eq: MeanAbsDiff}
\end{equation}

It is a classical result that any minimizer of~\eqref{Eq: MeanAbsDiff} is a sample median of the observations $\{y_i\}$. This follows from a subgradient argument. The subgradient of $|y_i - a|$ with respect to $a$ is $-\mathrm{sgn}(y_i - a)$, where $\mathrm{sgn}(\cdot)$ denotes the sign function. Consequently, the subdifferential of $\mathcal{L}(a)$ is
\[
\partial \mathcal{L}(a) = \sum_{i=1}^{n} -\mathrm{sgn}(y_i - a).
\]
A necessary and sufficient condition for optimality is that $0 \in \partial \mathcal{L}(a)$, which holds precisely when the number of observations greater than $a$ is at least as large as the number less than $a$, and vice versa. This characterizes the sample median \citep{koenker1978}.

\subsection{Minimizing Deviations for a No-Intercept Model}
\label{Sec: MinNoInt}

The second subproblem concerns the estimation of a single slope coefficient $b$ in a linear regression model without an intercept. The objective is
\begin{equation}
\mathcal{L}(b) = \sum_{i=1}^{n} |y_i - b x_i|.
\label{Eq: MeanAbsErr}
\end{equation}
Assuming $x_i \neq 0$, the objective can be rewritten as
\begin{equation}
\mathcal{L}(b) = \sum_{i=1}^{n} |x_i| \left| \frac{y_i}{x_i} - b \right|.
\end{equation}
(Observations with $x_i = 0$ contribute a constant term to the objective and do not affect the minimizer.)

This representation shows that minimizing~\eqref{Eq: MeanAbsErr} is equivalent to computing a weighted median of the ratios $\{y_i / x_i\}$, with corresponding weights $|x_i|$. The weighted median is defined as any value $b$ such that the total weight of ratios less than $b$ and greater than $b$ are both at least half of the total weight. This result follows directly from properties of quantile regression \citep{koenker1978}. Together, these two univariate optimization results form the foundation of the coordinate-wise updates used in our algorithm.

\subsection{Genetic Algorithms for High-Dimensional Parameter Estimation}
\label{sec:background:ga}

In highly underdetermined regression problems ($ p> n$), where the number of predictors exceeds the number of observations, classical optimization methods for LAD regression can be difficult to apply reliably. In such settings, it is common to consider derivative-free global optimization techniques that remain well-defined even in the absence of strong convexity or full-rank design matrices. Genetic algorithms (GAs) provide one such approach and are therefore useful both as a baseline and as a source of high-quality initializations for local optimization methods \citep{Holland1975, Goldberg1989}.

In the regression context, the goal is to solve
\[
\widehat{\beta} = \arg\min_{\beta \in \mathbb{R}^p} 
\mathcal{L}(\beta), \qquad 
\mathcal{L}(\beta) := \frac{1}{n}\sum_{i=1}^{n} \rho \bigl(y_i - x_i^\top \beta\bigr),
\]
where $\rho(\cdot)$ denotes a loss function, such as $\rho(u)=|u|$ for LAD regression.  

A genetic algorithm maintains a population
\(
\mathcal{P}^{(t)} = \{\beta^{(t)}_1,\dots,\beta^{(t)}_M\}
\)
of $M$ candidate solutions at generation $t$, each associated with a fitness value
\[
F(\beta) = - \mathcal{L}(\beta) - \lambda \|\beta_{-1}\|_2^2,
\]
where an optional $\ell_2$ penalty is included to stabilize the search in the $p>n$ regime by discouraging excessively large coefficients. Here $\beta_{-1}$ denotes the subvector of regression coefficients excluding the intercept term. At each generation, the population is updated through the repeated application of selection, crossover, and mutation operators:
\begin{itemize}
    \item Selection: parents are chosen with probability proportional to their fitness;
    \item Crossover: selected parents are recombined to produce offspring;
    \item Mutation: individual coordinates are randomly perturbed with a small probability.
\end{itemize}
The next generation $\mathcal{P}^{(t+1)}$ is formed by replacing low-fitness candidates with newly generated offspring, and the procedure continues until a stopping criterion is satisfied.

Because genetic algorithms do not rely on gradient information or matrix inversion, they remain applicable in underdetermined settings where the design matrix is rank-deficient. This generality, however, comes at a substantial computational cost: the per-generation complexity scales as $O(M n p)$, and without explicit diversity-preserving mechanisms, the population may converge prematurely to suboptimal solutions \citep{DeJong1975, Vose1999}.

Accordingly, GAs are most effective when combined with fast local refinement methods. In this work, we use a GA to provide an initial estimate, which is then refined using the proposed coordinate descent algorithm, thereby combining global exploration with efficient, provably convergent local optimization.

\section{Proposed Methodology}
\label{sec:methodology}

In this section, we present a coordinate descent algorithm for solving the Least Absolute Deviations (LAD) regression problem. We begin by describing a naive coordinate descent formulation, which highlights the structure of the LAD objective and its coordinate-wise minimizers. We also provide a proof of convergence of this method. Subsequent subsections introduce computational improvements and practical refinements. We also discuss strategies for initializing the parameter estimates and compare their effectiveness in different settings.

\subsection{Na\"ive Coordinate Descent for LAD}
\label{sec:methodology: naive}

Recall that the LAD estimator is defined as
\[
\widehat{\beta}_{\text{LAD}}
= \arg\min_{\beta \in \mathbb{R}^p}
\mathcal{L}(\beta), \qquad
\mathcal{L}(\beta) := \sum_{i=1}^{n} \big| y_i - x_i^\top \beta \big|.
\]
Coordinate descent proceeds by iteratively updating one component of $\beta$ at a time while holding the remaining coordinates fixed. Let $\beta^{(t)}$ denote the parameter vector at iteration $t$. For a given coordinate $j$, define the partial residual
\begin{equation}
    \label{eq. residuals_naive_lad}
    r^{(j)}_i = y_i - \sum_{k \neq j} x_{ik} \beta^{(t)}_k.
\end{equation}
We use $r^{(j)} \in \mathbb{R}^n$ to denote the vector of partial residuals associated with coordinate $j$, with components $r_i^{(j)}$ for $i = 1, \dots, n$. Updating the $j$th coordinate reduces to solving the univariate optimization problem
\[
\beta_j^{(t+1)} = \arg\min_{\theta \in \mathbb{R}}
\sum_{i=1}^{n} \big| r^{(j)}_i - x_{ij} \theta \big|.
\]
As shown in Section~\ref{sec:background}, this problem admits a closed-form solution. When $x_{ij} \neq 0$ for all $i$, the minimizer is given by the weighted median of the ratios $z_i = r^{(j)}_i / x_{ij}$ with weights $w_i = |x_{ij}|$ (see Section~\ref{Sec: MinNoInt}):
\begin{align*}
\beta_j^{(t+1)}
&= \operatorname{wmed}(z_1,\dots,z_n; w_1,\dots,w_n) \\
&= \arg\min_{b \in \mathbb{R}} \sum_{i=1}^{n} w_i |z_i - b| \\
&= \left\{\, b \in \mathbb{R} \;\middle|\;
\sum_{i : z_i \le b} w_i \ge \tfrac{1}{2}\sum_{i=1}^{n} w_i
\;\text{ and}\;
\sum_{i : z_i \ge b} w_i \ge \tfrac{1}{2}\sum_{i=1}^{n} w_i
\right\}.
\end{align*}

In the special case where $x_{ij} = 1$ for all $i$ (corresponding to an intercept term), the update reduces to the ordinary median of $\{r^{(j)}_i\}$, as discussed in Section~\ref{sec:background}. A full iteration consists of cycling through all coordinates, and the procedure is repeated until convergence. Since $\mathcal{L}(\beta)$ is convex and each coordinate update is an exact minimization, the resulting algorithm is guaranteed to converge to a global minimizer. We clearly lay out the algorithm structure in Algorithm~\ref{alg:LAD_naive}.

\begin{algorithm}[!ht]
\caption{Naive Coordinate Descent for LAD}
\label{alg:LAD_naive}
\begin{algorithmic}[1]
\Require Design matrix \(X \in \mathbb{R}^{n \times p}\), response vector \(y \in \mathbb{R}^n\), initial vector \(\beta^{(0)}\), maximum iterations \(T\), tolerance \(\varepsilon\)
\State \(\beta \gets \beta^{(0)}\)
\For{\(t = 0, 1, \dots, T-1\)}
  \For{\(j = 1, \dots, p\)}
    \State Compute partial residual $r^{(j)}$ as defined in \eqref{eq. residuals_naive_lad}

    \If{\(j\) corresponds to intercept column}
        \State \(\beta_j \gets \operatorname{median}(r^{(j)})\)
    \Else
        \State Compute \(z_i = r^{(j)}_i / x_{ij}\), \(w_i = |x_{ij}|\) for all \(i\)
        \State \(\beta_j \gets \operatorname{weightedMedian}(z, w)\)
    \EndIf
  \EndFor
  \State Compute objective \(\mathcal{L}(\beta)\)
  \If{relative decrease in \(\mathcal{L}\) below \(\varepsilon\)}
    \State \textbf{break}
  \EndIf
\EndFor
\State \Return final \(\beta\)
\end{algorithmic}
\end{algorithm}

\subsection{Descent and Global Convergence}

The descent and convergence properties of the proposed algorithm follow from two basic facts: (i) the LAD objective is convex, and (ii) each coordinate update exactly minimizes the objective with respect to the selected coordinate while holding all others fixed. Together, these properties imply that the objective value is monotonically non-increasing across iterations and that the algorithm converges to a global minimizer. We formalize this argument below.

\begin{lemma}[Convexity of the LAD objective]
\label{lem: Convex LAD}
Let
\[
\mathcal{L}(\beta) = \sum_{i=1}^n | y_i - x_i^\top \beta |, \qquad \beta \in \mathbb{R}^{p}.
\]
Then $L$ is convex in $\beta$.
\end{lemma}

\begin{corollary}[Convexity of coordinate subproblems]
Fixing all coordinates except $\beta_j$ yields the one-dimensional subproblem
\[
g_j(t) = \sum_{i=1}^n |\, r_i^{(j)} - x_{ij} t |, \qquad 
r_i^{(j)} = y_i - \beta_0 - \sum_{k \neq j} \beta_k x_{ik}.
\]
As $g_j$ is a sum of convex functions of $t$, it is convex, and its minimizer is given by the weighted median of $\{r_i^{(j)}/x_{ij}\}$ with weights $|x_{ij}|$.
\end{corollary}

\begin{theorem}[Global convergence of LAD-CD]
\label{THM: LAD_CD-Convergence}
Let $\beta^{k}$ denote the parameter vector at iteration $k$, and let $\mathcal{L}(\beta^{k})$ be the corresponding objective value. Then
\[
\mathcal{L}(\beta^{k+1}) \leq \mathcal{L}(\beta^{k}), \quad \forall k,
\]
and the sequence $\{ \beta^{k} \}_{k=0}^{\infty}$ converges to a global minimizer $\beta^{\star}$ of $\mathcal{L}(\beta)$.
\end{theorem}

\noindent
Proofs for Lemma~\ref{lem: Convex LAD}, the corollary, and Theorem~\ref{THM: LAD_CD-Convergence} are deferred to Appendix~\ref{app:convergence}.

This result establishes that LAD-CD is both \emph{monotone} and \emph{globally convergent} under standard regularity conditions. Empirically, we observe convergence within a few dozen iterations in both low- and high-dimensional settings (see Section~\ref{sec:experiments}). 

While the naive coordinate descent formulation is straightforward and theoretically sound, it is computationally inefficient in practice because it requires recomputing partial residuals at each coordinate update. This cost becomes prohibitive in high-dimensional settings. We therefore introduce an optimized implementation that exploits incremental residual updates to eliminate redundant computations. Subsequent subsections build on this formulation by incorporating warm-start strategies and a hybrid global--local optimization approach.

\subsection{Optimized Coordinate Descent: Incremental Residual Updates}

A direct implementation of coordinate descent for LAD regression, while conceptually simple, suffers from a significant computational bottleneck that limits its applicability in high-dimensional problems. The inefficiency arises from repeatedly recomputing the partial residual vector for each coordinate update.

In the naive implementation, a full cycle over all $p$ coordinates constitutes one iteration. Let $X_{:j} \in \mathbb{R}^n$ denote the $j$th column of the design matrix $X$. Within this cycle, the partial residual is recomputed from scratch for each coordinate,
\[
r^{(j)} = y - \sum_{k \neq j} X_{:k} \beta_k,
\]
which requires a matrix--vector product of cost $O(np)$. Repeating this operation for all $p$ coordinates yields a per-iteration complexity of $O(np^2)$, which quickly becomes prohibitive as $p$ grows.

To overcome this limitation, we introduce an optimized implementation that maintains a single global residual vector,
\[
\mathrm{res} = y - X\beta,
\]
and updates it incrementally with each coordinate update. Rather than recomputing partial residuals from scratch, each coordinate update proceeds through three inexpensive steps.

\begin{description}
    \item[\textbf{1. Restore contribution ($O(n)$).}]
    The contribution of the current coefficient $\beta_j$ is temporarily restored to obtain the partial residual required for the update:
    \[
    r^{(j)} = \mathrm{res} + X_{:j}\beta_j.
    \]

    \item[\textbf{2. Compute coordinate update ($O(n \log n)$).}]
    The updated coefficient $\beta_j^{\text{new}}$ is obtained by solving the one-dimensional subproblem, which reduces to a weighted median computation:
    \[
    \beta_j^{\text{new}} = \operatorname{weightedMedian}\!\biggl( \frac{r^{(j)}_i}{x_{ij}}, \; |x_{ij}| \biggr).
    \]
    Using a standard sort-based routine, this step has complexity $O(n \log n)$.

    \item[\textbf{3. In-place residual update ($O(n)$).}]
    The global residual is updated using the coefficient change $\Delta_j = \beta_j^{\text{new}} - \beta_j$:
    \[
    \mathrm{res} \leftarrow \mathrm{res} - X_{:j}\Delta_j.
    \]
\end{description}

By replacing the $O(np)$ recomputation of partial residuals with $O(n)$ incremental updates, the cost of a single coordinate update is dominated by the weighted median computation. It is therefore $O(n \log n)$. As a result, a full sweep over all $p$ coordinates has complexity $O(p n \log n)$, representing a substantial improvement over the $O(np^2)$ cost of the naive implementation, particularly in high-dimensional settings.

In addition to this asymptotic speedup, the incremental formulation improves practical performance by enhancing memory locality. Each update accesses only a single column of $X$ (the data matrix) and the residual vector, leading to improved cache utilization and reduced wall-clock time. We lay this down clearly in Algorithm~\ref{alg:LAD_optimized}.

\begin{algorithm}[!ht]
\caption{Optimized LAD Coordinate Descent with Incremental Residual Updates}
\label{alg:LAD_optimized}
\begin{algorithmic}[1]
\Require Design \(X \in \mathbb{R}^{n \times p}\), response \(y \in \mathbb{R}^n\), initial \(\beta^{(0)}\), max iterations \(T\), tolerance \(\varepsilon\)
\State \(\beta \gets \beta^{(0)}\)
\State \(\mathrm{res} \gets y - X\beta\)
\For{\(t = 0, 1, \dots, T-1\)}
  \For{\(j = 1, \dots, p\)}
    \State \(r^{(j)} \gets \mathrm{res} + X_{:j} \beta_j\)
    \If{\(j\) corresponds to intercept column}
        \State \(\beta_j^{\text{new}} \gets \operatorname{median}(r^{(j)})\)
    \Else
        \State Compute \(z_i = r^{(j)}_i / x_{ij}\), \(w_i = |x_{ij}|\)
        \State \(\beta_j^{\text{new}} \gets \operatorname{weightedMedian}(z, w)\)
    \EndIf
    \State \(\Delta_j \gets \beta_j^{\text{new}} - \beta_j\)
    \State \(\beta_j \gets \beta_j^{\text{new}}\)
    \State \(\mathrm{res} \gets \mathrm{res} - X_{:j} \Delta_j\)
  \EndFor
  \State Compute objective $\mathcal{L}(\beta) = \sum_{i=1}^n |\mathrm{res}_i|$

  \If{relative decrease in \(\mathcal{L}\) below \(\varepsilon\)}
    \State \textbf{break}
  \EndIf
\EndFor
\State \Return final \(\beta\)
\end{algorithmic}
\end{algorithm}

\subsection{Warm-Start Strategies for Coordinate Descent}
\label{sec:method:warmstart}

Because coordinate descent is a descent-based optimization method, its convergence behavior, and in finite samples, even the attained solution, can depend on the choice of initialization. Selecting an appropriate starting point can therefore substantially accelerate convergence and improve finite-sample performance, particularly in high-dimensional settings or in the presence of outliers. We consider two complementary classes of warm-start strategies.

\paragraph{(i) Pre-fitted model initialization.}
When a computationally inexpensive preliminary estimator is available, LAD-CD can be initialized from its coefficient estimates. Two particularly effective choices are:
\begin{itemize}
    \item \textbf{Ridge regression.}  
    As a simple and computationally inexpensive pre-fit, we consider ridge regression, which solves a penalized least-squares problem of the form
    \[
    \widehat{\beta}_{\text{ridge}}
    = \arg\min_{\beta \in \mathbb{R}^p}
    \left\{ \sum_{i=1}^n (y_i - x_i^\top \beta)^2 + \lambda \|\beta\|_2^2 \right\},
    \]
    where $\lambda > 0$ is a regularization parameter that controls the strength of $\ell_2$ shrinkage. The corresponding solution can be written in closed form as
    \[
    \widehat{\beta}_{\text{ridge}}
    = \bigl(X^\top X + \lambda I\bigr)^{-1} X^\top y,
    \]
    where, again, $X \in \mathbb{R}^{n \times p}$ denotes the design matrix with rows $x_i^\top$ and $y \in \mathbb{R}^n$ is the response vector.
    
    The $\ell_2$ penalty stabilizes estimation in high-dimensional or ill-conditioned settings by shrinking coefficients toward zero, thereby reducing variance at the cost of a small bias \citep{HoerlKennard1970, hastie2009}. In practice, ridge regression provides a deterministic and numerically stable initialization, particularly when $p > n$. When used as a warm start, LAD-CD primarily acts to correct the bias induced by the quadratic loss and the $\ell_2$ penalty, while retaining the robustness of the LAD objective.

    \item \textbf{Genetic algorithm (GA).}  
    A GA-based initializer performs a global search over the parameter space and can explore multiple basins of attraction. Although computationally more expensive than ridge, this approach can be beneficial in settings with heavy-tailed noise or complex loss landscapes, where simple linear shrinkage may be inadequate.
\end{itemize}
Empirically, both ridge- and GA-based initializations lead to faster convergence and improved prediction accuracy compared to arbitrary starting values. In particular, the hybrid GA+CD approach can be useful in severely underdetermined settings ($p \ge n$), where LP-based LAD solvers may fail to produce reliable solutions.

\paragraph{(ii) Random multi-start initialization.}
Even in the absence of a pre-fitted model, LAD-CD typically converges in a small number of iterations—often stabilizing within $30$--$50$ iterations (see Section~\ref{sec:experiments}). This low per-iteration cost enables a simple multi-start strategy: multiple random initializations $\{\beta^{(0)}_k\}_{k=1}^K$ are generated, the algorithm is run from each starting point, and the solution with the smallest final loss is retained. This approach provides an inexpensive way to reduce sensitivity to initialization by exploring multiple regions of the parameter space.

\medskip
In practice, we recommend a two-tiered strategy: use a pre-fitted model, particularly ridge regression, whenever available to achieve fast, stable convergence, and employ random multi-starts as a fallback when such initializations are unavailable or when additional robustness is desired.

\section{Experiments}
\label{sec:experiments}

This section evaluates the empirical performance of the proposed LAD coordinate descent algorithm. We first study its convergence behavior and stability under controlled synthetic settings, where ground-truth parameters are known, and then assess robustness across repeated random initializations. Subsequent experiments consider contaminated data, high-dimensional regimes, and real-world benchmark datasets.

To evaluate the proposed method in a controlled setting, we generate synthetic data from the linear model
\[
y = X \beta^\star + \varepsilon,
\qquad \varepsilon \sim \mathcal{N}(0, \sigma^2 I_n),
\]
where $X \in \mathbb{R}^{n \times p}$ has i.i.d.\ standard normal entries, with an intercept column included, and $\beta^\star$ is a fixed coefficient vector. This setup allows direct assessment of parameter recovery and predictive accuracy, and enables comparison across both well-posed ($p < n$) and underdetermined ($p \ge n$) regimes (we also talk about the latter in a later experiment).

We apply the proposed LAD coordinate descent algorithm to this dataset and track its behavior over iterations. Figure~\ref{fig: EM LAD generated data set} illustrates the fitted regression line obtained by the algorithm, along with the evolution of the mean absolute error (MAE) for both parameter estimates and predictions. These diagnostics provide insight into the method's convergence dynamics and its ability to recover the underlying signal.

\begin{figure}[!ht]
    \centering
    \includegraphics[width=0.45\linewidth]{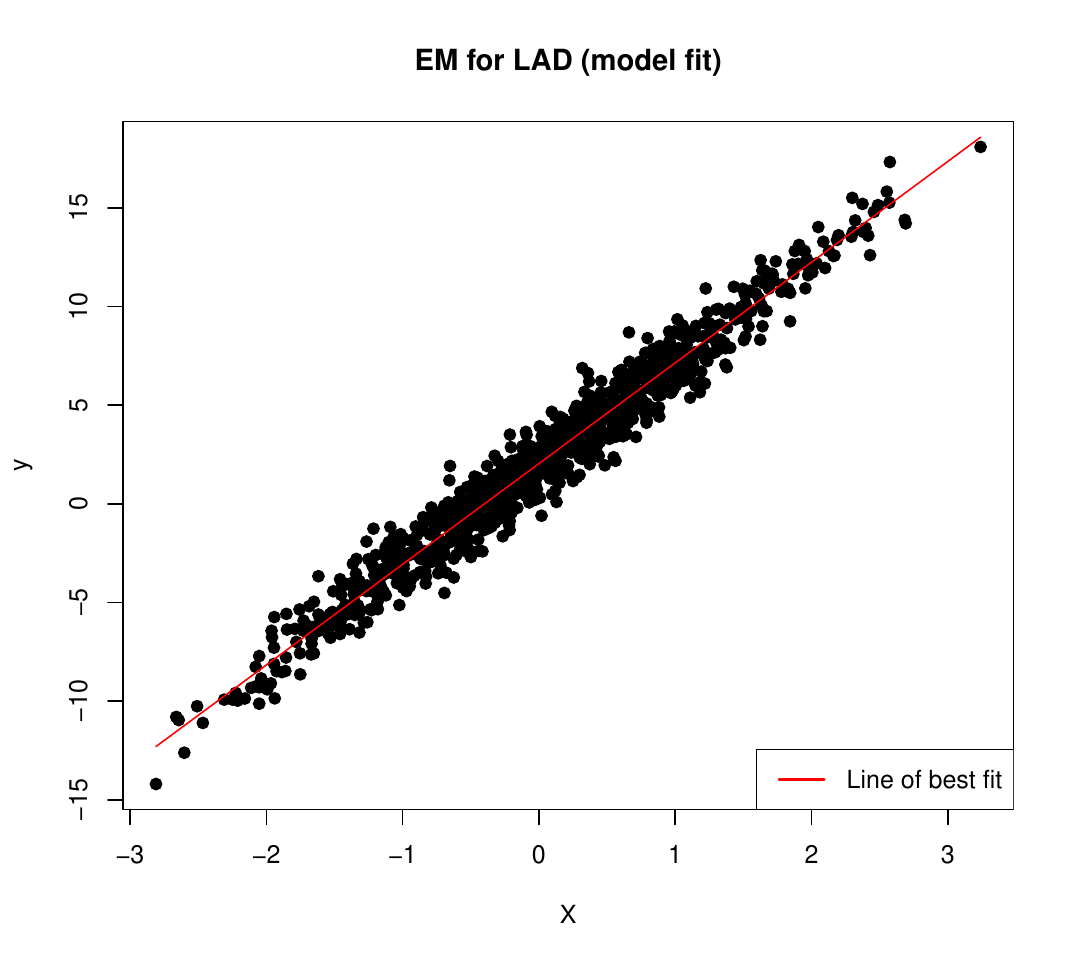}
    \hfill
    \includegraphics[width=0.45\linewidth]{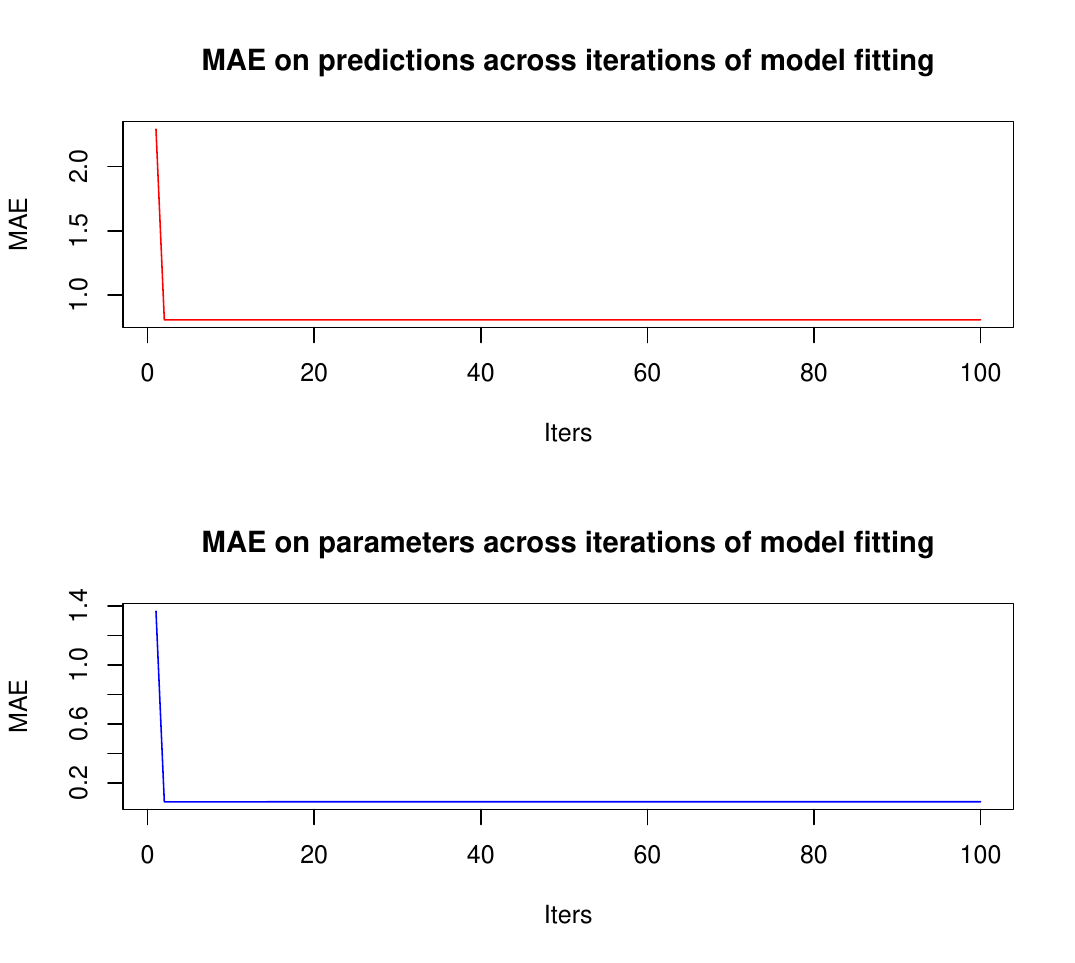}
    \caption{performance of the LAD coordinate descent algorithm on synthetic data. 
    Left: fitted regression line obtained by LAD-CD. 
    Right: evolution of parameter MAE and prediction MAE over iterations.}
    \label{fig: EM LAD generated data set}
\end{figure}

To assess robustness and stability with respect to initialization, we repeat the synthetic experiment described above 1000 times, each with a different random seed and random initialization of the coefficient vector. For each run, the MAE is recorded at every iteration until convergence. This allows us to examine both the average convergence trajectory and the variability of the final solution across repeated runs.

Figure~\ref{fig: LADCD_convergence_distribution} summarizes the distribution of final MAE values across the 1000 replications. The distribution is sharply concentrated, with a mean final MAE of approximately $0.7539$ and a standard deviation of $6 \times 10^{-4}$. The corresponding 90\% confidence interval, $[0.7529, 0.7555]$, is narrow, indicating that the optimized LAD-CD algorithm converges to essentially identical solutions across runs, despite random initialization.

We further examine the final prediction MAE obtained by fitting the synthetic model described above using the proposed method for varying sample sizes $n$, with the dimension fixed at $p = 100$. We compare two initialization strategies: (i) an uninformed start with all coefficients initialized at zero, and (ii) a warm start obtained by perturbing the true parameter vector with independent $\mathcal{N}(0,1)$ noise. The resulting MAE values are reported in Table~\ref{tab:initial_mae_scaling}.

When $n$ is very small relative to $p$ (e.g., $n \in \{10, 20\}$), the problem is severely underdetermined, and the fitted solutions exhibit substantial variability across initializations. In moderately underdetermined regimes ($n \in \{50, 100\}$), warm-start initialization leads to improved predictive performance and more stable parameter estimates, indicating that initialization plays an important role in finite samples. Finally, when the model becomes well-specified ($n \ge p$), the algorithm consistently attains stable MAE values across both initialization strategies, reflecting the underlying convex structure of the objective.

\begin{figure}[!ht]
\centering
\includegraphics[width=0.6\linewidth]{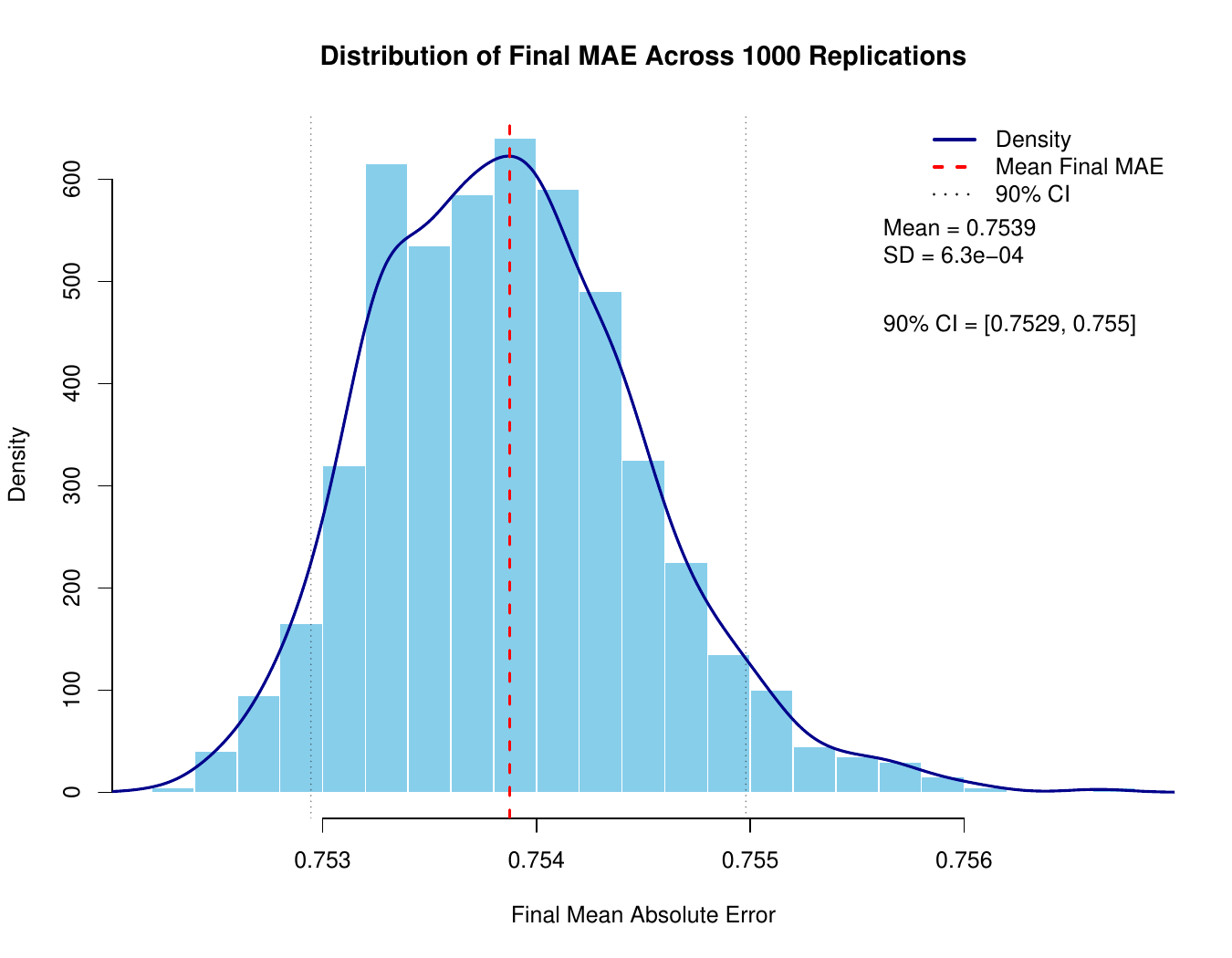}
\caption{Distribution of final prediction MAE across 1000 independent replications, illustrating the convergence stability of the optimized LAD-CD algorithm.}
\label{fig: LADCD_convergence_distribution}
\end{figure}

\subsection{Performance on Data with Outliers}
\label{sec:experiments:outliers}

We first evaluate the robustness of the proposed LAD-CD algorithm under contamination by outliers. A simple linear regression dataset is generated with $n = 1000$ observations and a single covariate $x$, according to
\[
y_i = \beta_0 + \beta_1 x_i + \varepsilon_i,
\qquad \beta_0 = 2,\ \beta_1 = 5,\ \varepsilon_i \sim \mathcal{N}(0, 5^2).
\]
To introduce contamination, $20\%$ of the observations are replaced with outliers generated from the same model but with inflated noise variance, $\varepsilon_i \sim \mathcal{N}(0, 25^2)$.

Three estimators are fit to this dataset:
\begin{enumerate}
    \item LAD regression using the proposed coordinate descent algorithm (LAD-CD);
    \item Ordinary Least Squares (OLS) regression as a baseline;
    \item Median (quantile) regression at $\tau = 0.5$ using the \texttt{quantreg} package, which solves the LAD problem via linear programming.
\end{enumerate}

Figure~\ref{fig:outlier-fits} compares the fitted regression lines obtained by the three methods. The OLS fit is clearly influenced by the contaminated observations, leading to a biased estimate of the underlying linear trend. In contrast, both LAD-CD and QuantReg closely recover the true signal, reflecting the robustness of the LAD objective.

To examine convergence behavior, we track the mean absolute error (MAE) of both parameter estimates and predictions across LAD-CD iterations. These trajectories are shown alongside horizontal reference lines indicating the MAEs achieved by OLS and QuantReg. The LAD-CD algorithm converges within approximately 20 iterations and attains parameter and prediction errors comparable to those of QuantReg, while OLS exhibits substantially larger errors.

\begin{figure}[!ht]
    \centering
    \includegraphics[width=0.45\linewidth]{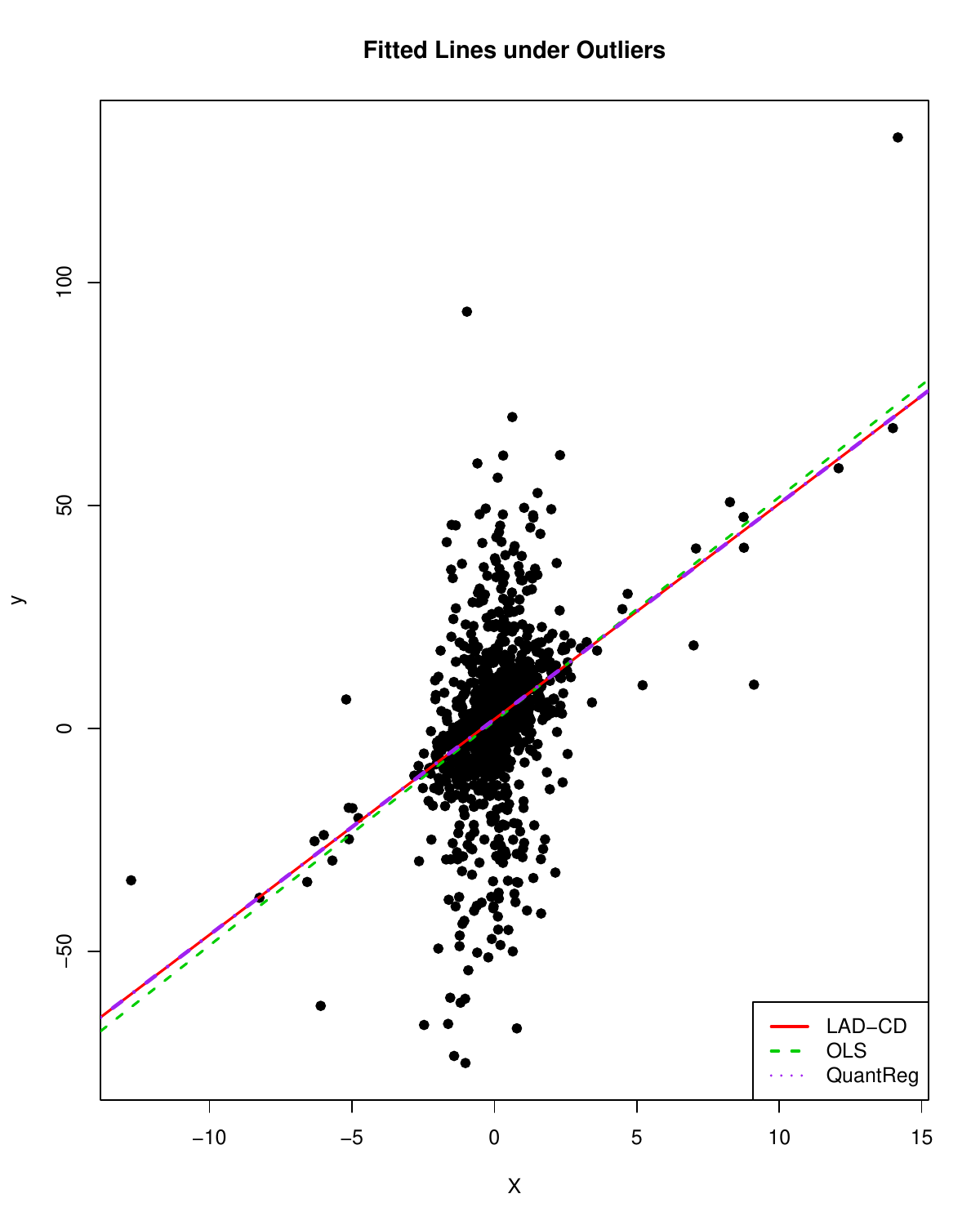}
    \hfill
    \includegraphics[width=0.45\linewidth]{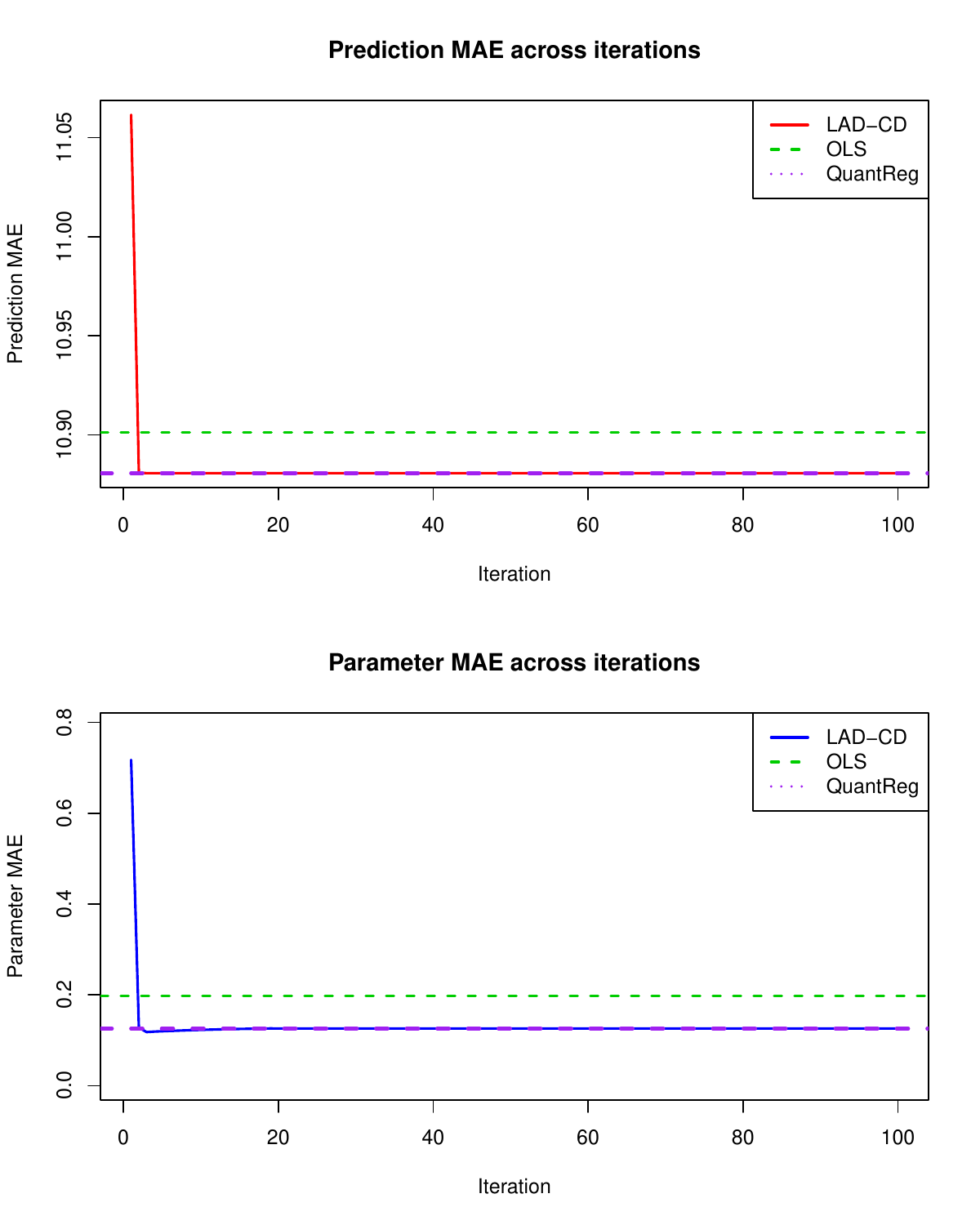}
    \caption{Fitted regression lines for LAD-CD (red), OLS (green dashed), and QuantReg (purple dotted) on the outlier-contaminated dataset (left), and convergence of parameter and prediction MAE for LAD-CD across iterations (right).}
    \label{fig:outlier-fits}
\end{figure}

\subsection{Improving Pre-Fitted Models via LAD-CD}
\label{sec:experiments:warmstart}

We next examine the effectiveness of the proposed warm-start strategies by assessing whether LAD-CD can refine parameter estimates obtained from simpler pre-fitted models. Two initializations are considered: (i) a closed-form ridge regression estimator and (ii) a genetic algorithm (GA)- based global search solution. Both initializations are applied to a synthetic dataset with heavy-tailed noise and $20\%$ severe contamination in the response variable.

Figure~\ref{fig:warmstart-combined} presents a visual comparison of the fitted models. The left panel shows the initial fits obtained from ridge regression and the GA, while the right panel displays the corresponding fits after LAD-CD refinement. In both cases, the refinement step substantially reduces the influence of outliers, yielding fitted lines that more closely track the underlying data-generating process.

Quantitative results are reported in Table~\ref{tab:warmstart-mae}, which summarizes the prediction mean absolute error (MAE) before and after LAD-CD refinement, along with the best achievable MAE computed using the true generating parameters $\beta_{\text{true}}$. When initialized from ridge regression, LAD-CD reduces prediction MAE by approximately $32.5\%$, while GA-based initialization yields a larger reduction of approximately $44\%$. In both cases, the refined solutions attain MAE values essentially indistinguishable from the best achievable benchmark, indicating that LAD-CD effectively corrects the bias present in the initial estimators even under substantial contamination.

\begin{figure}[!ht]
    \centering
    \includegraphics[width=0.6\linewidth]{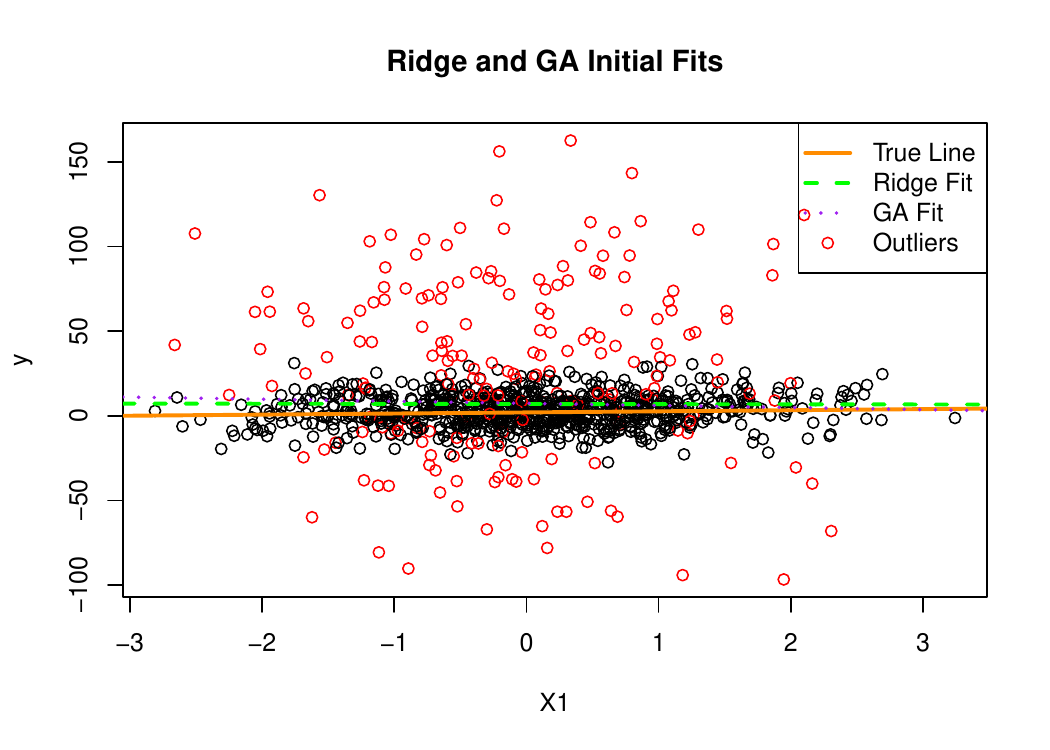}
    \hfill
    \includegraphics[width=0.35\linewidth]{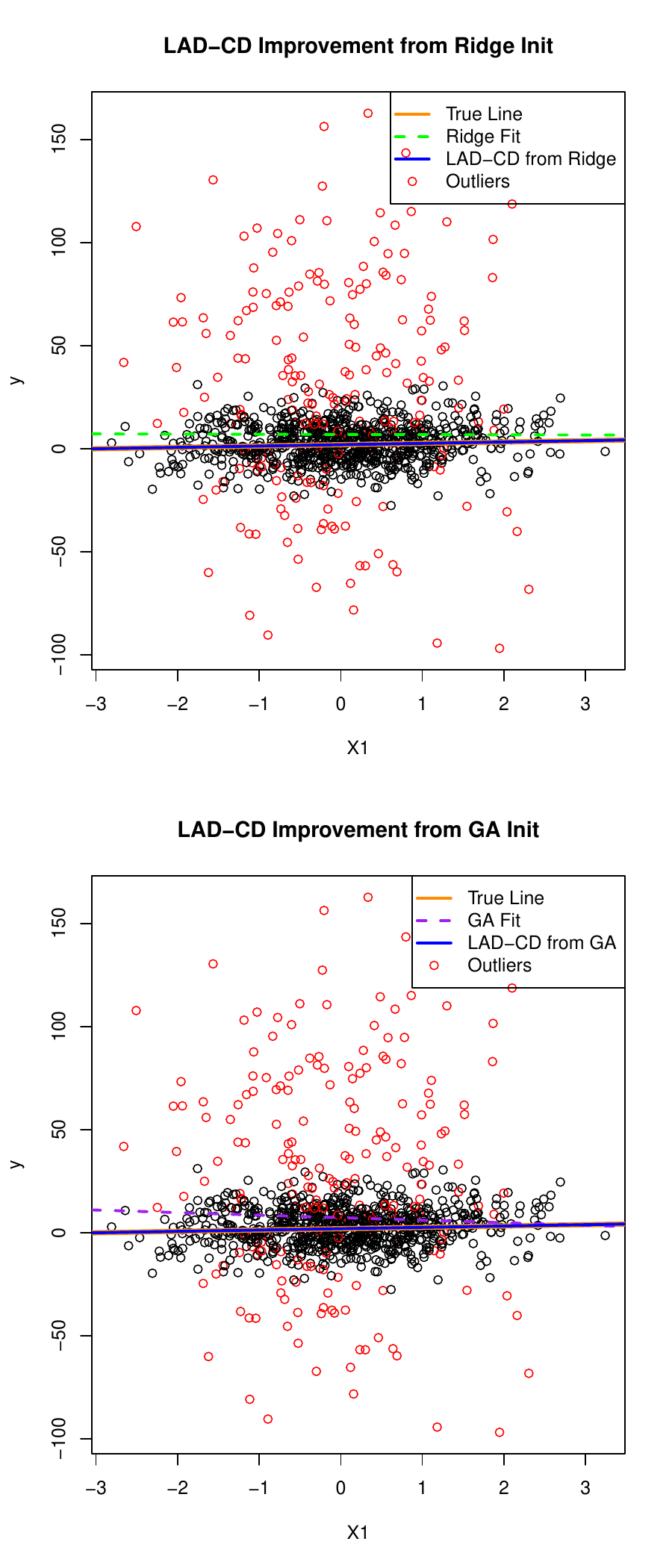}
    \caption{Warm-start results on outlier-contaminated data. 
    \textbf{Left:} Initial Ridge and GA fits. 
    \textbf{Right:} LAD–CD refinement substantially reduces bias and moves the fitted line closer to the true signal.}
    \label{fig:warmstart-combined}
\end{figure}

\begin{table}[!ht]
\centering
\begin{tabular}{lcccc}
\hline
\textbf{Initialization} & \textbf{Initial MAE} & \textbf{Final MAE} & \textbf{Best MAE} & \textbf{Improvement} \\
\hline
Ridge Init + LAD–CD & 13.49 & 9.10 & 9.10 & 4.39 \\
GA Init + LAD–CD    & 16.29 & 9.10 & 9.10 & 7.19 \\
\hline
\end{tabular}
\caption{Prediction MAE before and after LAD–CD refinement on outlier-contaminated data. The final MAE is nearly identical to the best achievable MAE, computed using the true generating parameters.}
\label{tab:warmstart-mae}
\end{table}

\subsection{Performance in High-Dimensional Settings}
\label{sec:experiments:highdim}

We next evaluate the performance of the proposed method in high-dimensional regimes. Synthetic datasets are generated with $n = 1000$ observations and $p \in \{100, 500, 1000, 1500, 2000\}$ predictors. For each dataset, we compare the following estimators:
\begin{enumerate}
    \item \textbf{QuantReg (QR):} Median regression solved via the Barrodale--Roberts linear programming algorithm;
    \item \textbf{LAD-CD (zero-init):} Coordinate descent initialized at $\beta^{(0)} = 0$;
    \item \textbf{Ridge regression:} $\ell_2$-regularized least squares;
    \item \textbf{Genetic algorithm (GA):} Global search via a GA metaheuristic;
    \item \textbf{Ridge + LAD-CD:} Coordinate descent initialized at the ridge solution;
    \item \textbf{GA + LAD-CD:} Coordinate descent initialized at the GA solution;
    \item \textbf{Best MAE:} Oracle reference computed as $\mathrm{MAE}(y - X\beta_{\text{true}})$.
\end{enumerate}
Results are summarized in two complementary tables: Table~\ref{tab:highdim-ridge} reports ridge-based initialization, while Table~\ref{tab:highdim-ga} reports GA-based initialization.

\begin{table}[H]
\centering
\caption{Performance with Ridge initialization. Best MAE is computed as $\text{MAE}(y - X\beta_{\text{true}})$. Runtime reports the total time for Ridge fitting + LAD-CD refinement.}
\label{tab:highdim-ridge}
\scriptsize
\begin{tabular}{rrrrrrrr}
\toprule
$n$ & $p$ & MAE(QR) & MAE(LAD-CD) & MAE(Ridge) & MAE(LAD-CD$\mid$Ridge) & Best MAE & Runtime\\
\midrule
1000 &  100 & 0.739 & 0.743 & 0.759 & \textbf{0.741} & 0.819 & 2.50 \\
1000 &  500 & 0.481 & 1.318 & 0.572 & \textbf{0.523} & 0.812 & 28.17 \\
1000 & 1000 & $\sim$ 0 & 4.925 & 0.157 & \textbf{0.137} & 0.795 & 107.59\\
1000 & 1500 & $\dagger$ & 6.735 & 0.034 & \textbf{0.023} & 0.785 & 228.66\\
1000 & 2000 & $\dagger$ & 7.984 & 0.026 & \textbf{0.015} & 0.785 & 397.78\\
\bottomrule
\end{tabular}
\end{table}

\begin{table}[H]
\centering
\caption{Performance with Genetic Algorithm (GA) initialization. Best MAE is computed as $\text{MAE}(y - X\beta_{\text{true}})$. Runtime reports GA search + LAD-CD refinement.}
\label{tab:highdim-ga}
\scriptsize
\begin{tabular}{rrrrrrrr}
\toprule
$n$ & $p$ & MAE(QR) & MAE(LAD-CD) & MAE(GA) & MAE(LAD-CD$\mid$GA) & Best MAE & Runtime\\
\midrule
1000 &  100 & 0.739 & 0.743 & 7.833 & 0.743 & 0.819 & 1.82\\
1000 &  500 & 0.481 & 1.318 & 21.824 & 1.754 & 0.812 & 28.55\\
1000 & 1000 & $\sim$ 0 & 4.925 & 30.142 & 6.736 & 0.795 & 107.51\\
1000 & 1500 & $\dagger$ & 6.735 & 39.830 & 8.977 & 0.785 & 225.15\\
1000 & 2000 & $\dagger$ & 7.984 & 48.353 & 11.531 & 0.785 & 390.584\\
\bottomrule
\end{tabular}
\end{table}

Tables~\ref{tab:highdim-ridge} and~\ref{tab:highdim-ga} reveal several consistent patterns. The symbol $\dagger$ indicates cases in which the \texttt{QuantReg} solver fails to return a solution due to underdetermination ($p > n$). When $p = n$, the quantile regression solver achieves a trivially small in-sample MAE due to exact interpolation, which does not reflect genuine predictive performance.

\begin{itemize}
    \item \textbf{Ridge-initialized LAD-CD performs best across dimensions.}  
    For all values of $p$, the combination of ridge initialization followed by LAD-CD refinement achieves the lowest prediction MAE. This indicates that a well-regularized initial estimate provides an effective starting point, allowing LAD-CD to focus on correcting the bias introduced by quadratic loss and $\ell_2$ shrinkage.

    \item \textbf{Ridge outperforms the unregularized oracle in the $p > n$ regime.}  
    When $p > n$, ridge regression yields lower prediction MAE than the loss computed using the true parameter vector $\beta_{\text{true}}$. This behavior reflects the classical bias--variance trade-off: ridge regression shrinks coefficient estimates toward zero, reducing variance and improving predictive accuracy in finite samples. This phenomenon is well documented in the ridge regression literature \citep{HoerlKennard1970} and is closely related to Stein-type shrinkage effects \citep{JamesStein1961, EfronMorris1977}; see also \citet{hastie2009} for a modern discussion.

    \item \textbf{GA warm-starts are less effective in high dimensions.}  
    While GA-based initialization produces feasible solutions even when $p \ge n$, the resulting estimates exhibit substantially larger prediction error. LAD-CD refinement improves these solutions but does not close the gap relative to ridge-based initialization. This suggests that GA is better suited as a fallback initialization strategy rather than a primary estimator in high-dimensional settings.

    \item \textbf{LP-based quantile regression does not scale to $p \ge n$.}  
    The quantile regression solver fails to return solutions in underdetermined regimes, whereas LAD-CD remains stable and well-defined. When $p < n$, quantile regression performs competitively, but its lack of scalability highlights a key advantage of the proposed approach.

    \item \textbf{Scalability and runtime.}  
    Although runtime increases with $p$, ridge-initialized LAD-CD remains computationally feasible even for $p = 2000$ with $n = 1000$, a regime that poses significant challenges for standard LAD solvers.
\end{itemize}

Overall, these results demonstrate that ridge-initialized LAD-CD provides a scalable and effective approach to robust regression in high-dimensional settings, achieving lower prediction error than GA-based alternatives and remaining applicable in regimes where LP-based LAD solvers fail.

\subsection{Boston Housing dataset}
\label{sec:experiments:boston}

We evaluate the proposed LAD coordinate descent algorithm on the Boston Housing dataset, a standard benchmark for robust regression. The dataset contains 506 observations with 13 covariates, and the response variable is the median house value (\texttt{medv}). Owing to the presence of outliers and heteroskedasticity, this dataset is commonly used to assess the performance of robust regression methods.

We compare LAD regression fitted using the proposed coordinate descent algorithm (LAD-CD) with the classical implementation provided by the \texttt{quantreg} package \citep{KoenkerDOrey1987}, which solves the LAD problem via the Barrodale--Roberts linear programming algorithm. Performance is assessed in terms of prediction accuracy, measured by mean absolute error (MAE), and computational runtime.

Figure~\ref{fig:boston_pred} (left) displays predicted house values against the observed responses for both methods. The fitted values are very similar, indicating that LAD-CD converges to solutions closely matching those obtained by the established LP-based solver. Figure~\ref{fig:boston_pred} (right) shows the evolution of the MAE across LAD-CD iterations. The loss decreases monotonically and stabilizes at a level close to the MAE achieved by \texttt{quantreg}, consistent with the theoretical convergence properties of the algorithm.

\begin{figure}[!ht]
    \centering
    \includegraphics[width=0.45\linewidth]{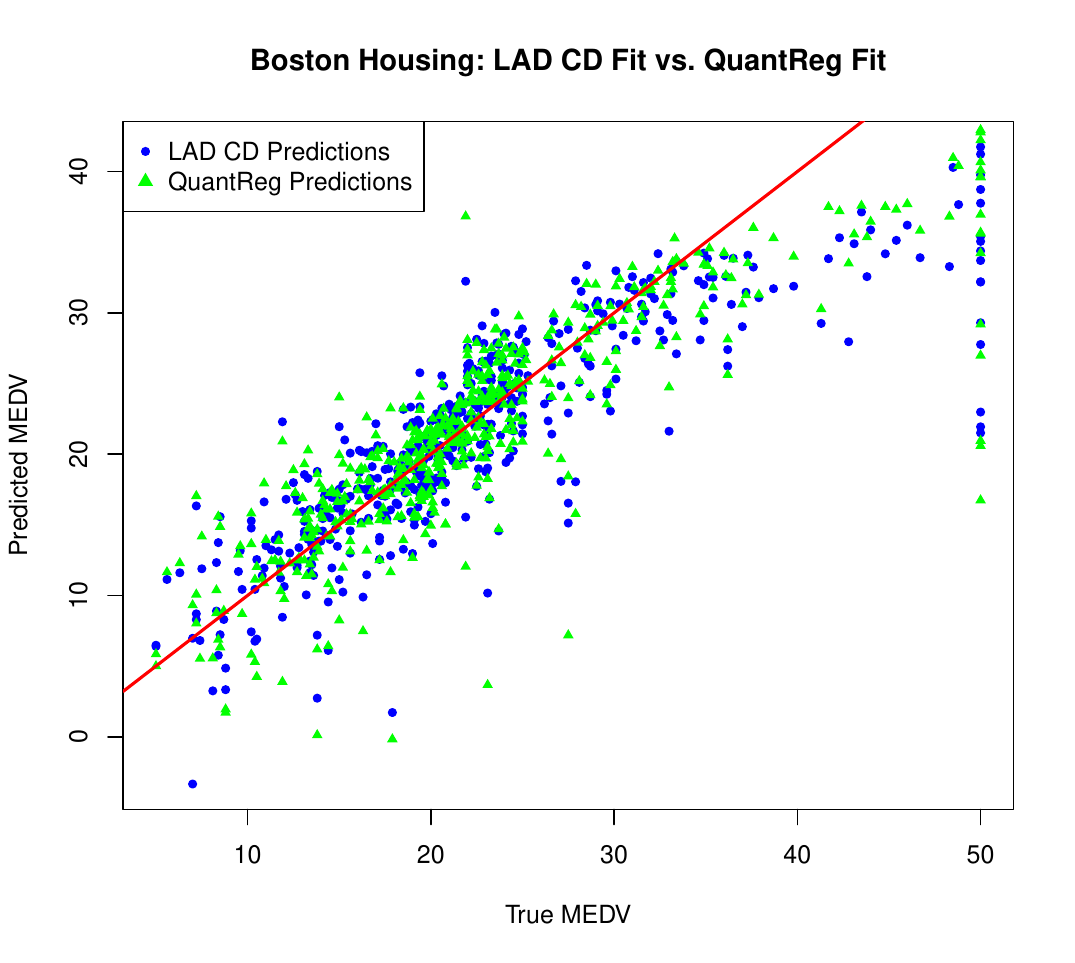}
    \hfill
    \includegraphics[width=0.45\linewidth]{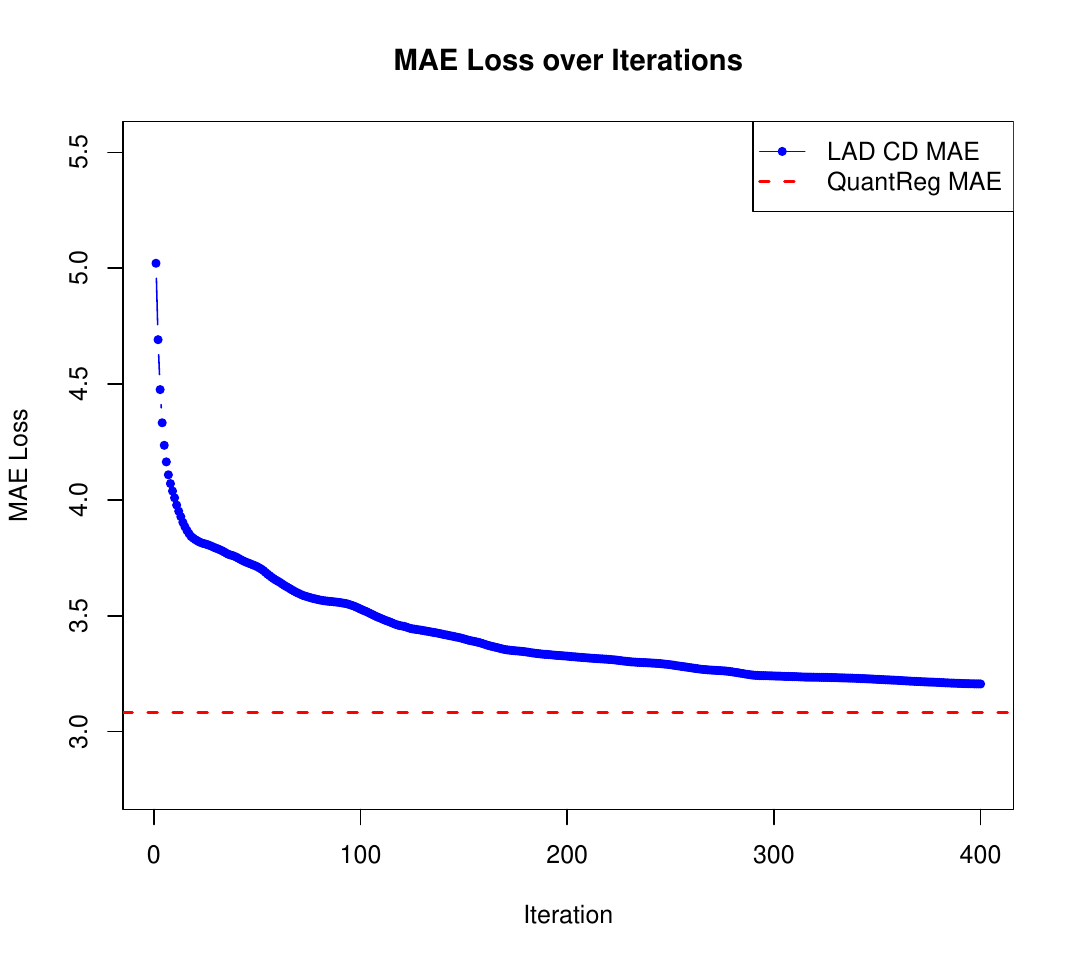}
    \caption{[Left] Predicted vs.\ true median house values for the Boston Housing dataset. Both LAD-CD (blue) and \texttt{quantreg} (green) produce similar fits, illustrating the consistency of our method. [Right] MAE loss over iterations for LAD-CD compared to \texttt{quantreg}. The coordinate descent method converges monotonically toward the optimal loss achieved by the LP-based solver.}
    \label{fig:boston_pred}
\end{figure}

Quantitative results are summarized in Table~\ref{tab:boston_comparison}. While the LP-based solver is faster on this relatively small dataset, LAD-CD achieves comparable predictive accuracy using a simple, matrix-inversion-free optimization procedure that does not rely on specialized linear programming routines.

\begin{table}[!ht]
\centering
\caption{Comparing (LAD-CD) and (QuantReg) on the Boston Housing dataset}
\label{tab:boston_comparison}
\begin{tabular}{|l|cc|c|}
\hline
 & MAE (loss) & Runtime (s) & $\sum_j |\hat\beta^{CD}_j - \hat\beta^{rq}_j|$ \\
\hline
LAD-CD & 3.53 & 0.115 & 30.24 \\
QuantReg (Barrodale--Roberts) & 3.08 & 0.002 & -- \\
\hline
\end{tabular}
\end{table}

\subsection{Air Quality: Comparison of LAD Coordinate Descent and Quantile Regression}
\label{sec:experiments:airquality}

We next evaluate the proposed method on the \texttt{airquality} dataset, which contains daily air quality measurements in New York, including ozone concentration, solar radiation, wind speed, and temperature. After removing observations with missing values, ozone concentration is modeled as the response variable using all available predictors. This dataset exhibits substantial variability and heavy-tailed residual behavior, making it a natural test case for LAD-based regression methods.

As in the Boston Housing experiment, we compare LAD regression fitted using the proposed coordinate descent algorithm (LAD-CD) with the classical implementation provided by the \texttt{quantreg} package, which solves the LAD problem via the \cite{BarrodaleRoberts1973} linear programming method. The primary goal of this experiment is to assess robustness and convergence behavior in a noisier real-world setting.

Figure~\ref{fig:air_pred} (left) compares predicted ozone concentrations from the two methods against the observed values. Despite noticeable scatter due to extreme observations, both approaches yield broadly similar predictions. Figure~\ref{fig:air_pred} (right) shows the evolution of the LAD-CD objective over iterations. The loss decreases monotonically and stabilizes at a level close to that achieved by the LP-based quantile regression solver.

\begin{figure}[!ht]
    \centering
    \includegraphics[width=0.45\linewidth]{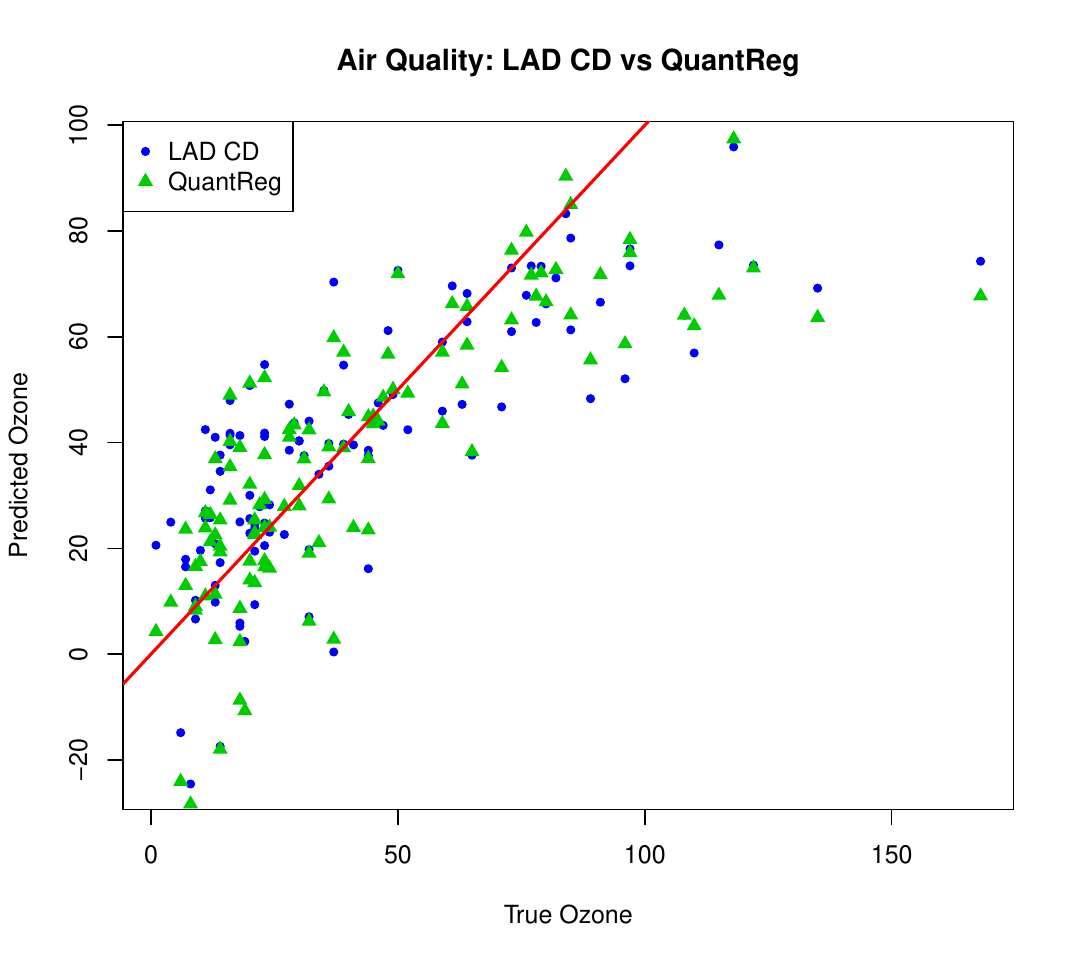}
    \hfill
    \includegraphics[width=0.45\linewidth]{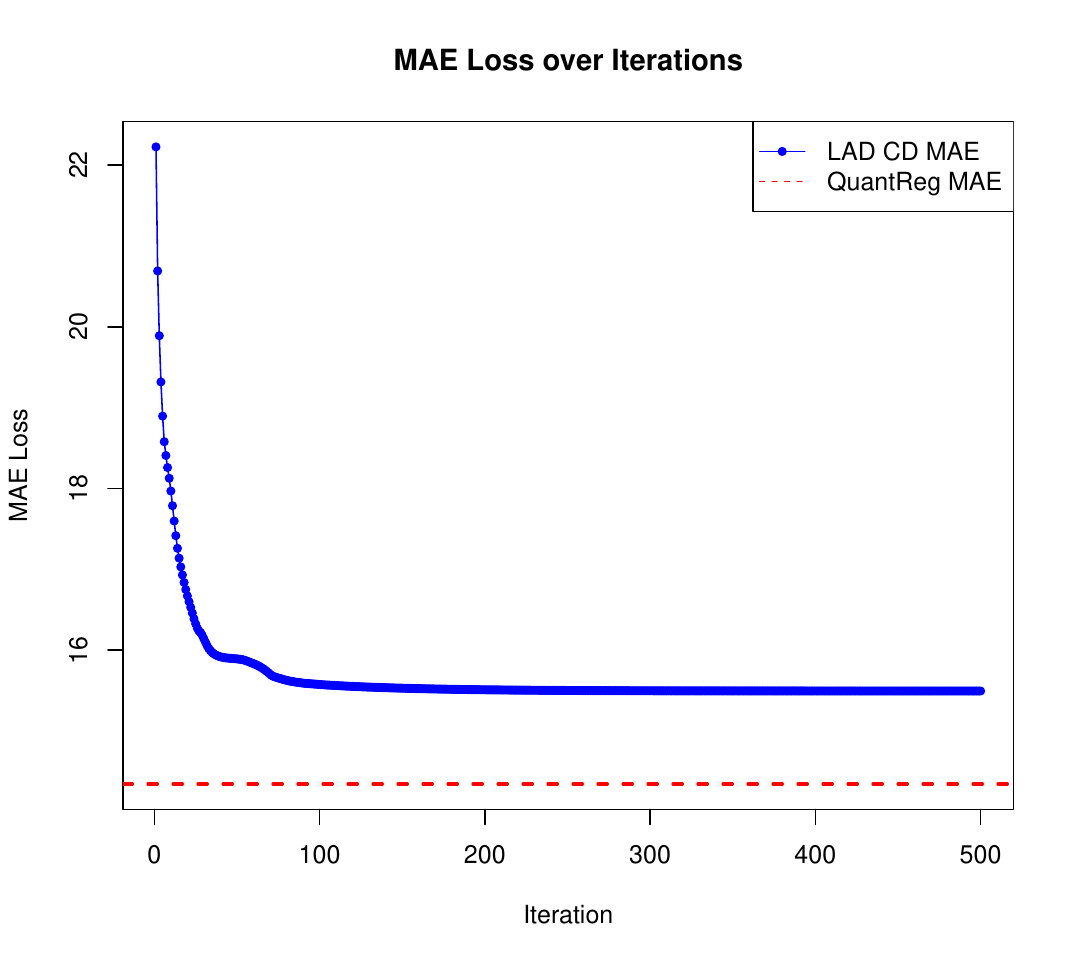}
    \caption{[Left] Predicted vs.\ true ozone levels for the Air Quality dataset. Both LAD-CD (blue) and QuantReg (green) yield broadly similar predictions, though the scatter reflects the heavy-tailed nature of the data. [Right] MAE loss over iterations for LAD-CD compared to QuantReg. The coordinate descent method converges monotonically, reaching a level close to the QuantReg baseline.}
    \label{fig:air_pred}
\end{figure}

Quantitative results are reported in Table~\ref{tab:airquality_comparison}. As in the Boston Housing dataset, LAD-CD attains predictive accuracy comparable to that of \texttt{quantreg}, albeit with higher computational cost on this relatively small dataset. The larger coefficient discrepancy observed between the two solutions reflects the increased variability and sensitivity to outliers in this dataset, though this difference does not materially affect predictive performance.

\begin{table}[!ht]
\centering
\caption{Comparison of LAD Coordinate Descent (LAD-CD) and Quantile Regression (QuantReg) on the Air Quality dataset.}
\label{tab:airquality_comparison}
\begin{tabular}{|l|cc|c|}
\hline
 & MAE (loss) & Runtime (s) & $\sum_j |\hat\beta^{CD}_j - \hat\beta^{rq}_j|$ \\
\hline
LAD-CD & 15.49 & 0.079 & 78.29 \\
QuantReg (Barrodale--Roberts) & 14.34 & 0.001 & -- \\
\hline
\end{tabular}
\end{table}

\subsection{Concrete Compressive Strength: Comparison of LAD Coordinate Descent and Quantile Regression}
\label{sec:experiments:concrete}

As a third real-world benchmark, we consider the \texttt{Concrete Compressive Strength} dataset from the UCI repository. The dataset consists of 1030 observations with eight covariates describing the composition of concrete mixtures, including cement, water, slag, fly ash, coarse and fine aggregates, superplasticizer, and curing age. The response variable is compressive strength, measured in MPa. Owing to heteroskedasticity and the presence of outliers, this dataset is commonly used to evaluate robust regression methods.

We compare LAD regression fitted using the proposed coordinate descent algorithm (LAD-CD) with the classical implementation provided by the \texttt{quantreg} package, which solves the LAD problem via the Barrodale--Roberts linear programming algorithm. Both methods yield similar predictive performance, with prediction MAEs of $8.14$ for LAD-CD and $8.05$ for QuantReg. As in previous experiments, the LP-based solver is faster on this moderate-sized dataset, while LAD-CD remains computationally feasible and straightforward to implement.

Table~\ref{tab:concrete_comparison} reports the numerical results, and Figure~\ref{fig:concrete_pred} illustrates predicted versus observed compressive strengths along with the convergence behavior of LAD-CD. Although the coefficient discrepancy between the two solutions is larger in this setting, reflecting increased variability in the predictors, the predictive behavior of the two methods remains closely aligned.

\begin{table}[H]
\centering
\caption{Comparison of LAD Coordinate Descent (LAD-CD) and Quantile Regression (QuantReg) on the Concrete Compressive Strength dataset.}
\label{tab:concrete_comparison}
\begin{tabular}{|l|cc|c|}
\hline
 & MAE (loss) & Runtime (s) & $\sum_j |\hat\beta^{CD}_j - \hat\beta^{rq}_j|$ \\
\hline
LAD-CD & 8.14 & 0.258 & 98.91 \\
QuantReg (Barrodale--Roberts) & 8.05 & 0.003 & -- \\
\hline
\end{tabular}
\end{table}

\begin{figure}[!ht]
    \centering
    \includegraphics[width=0.45\linewidth]{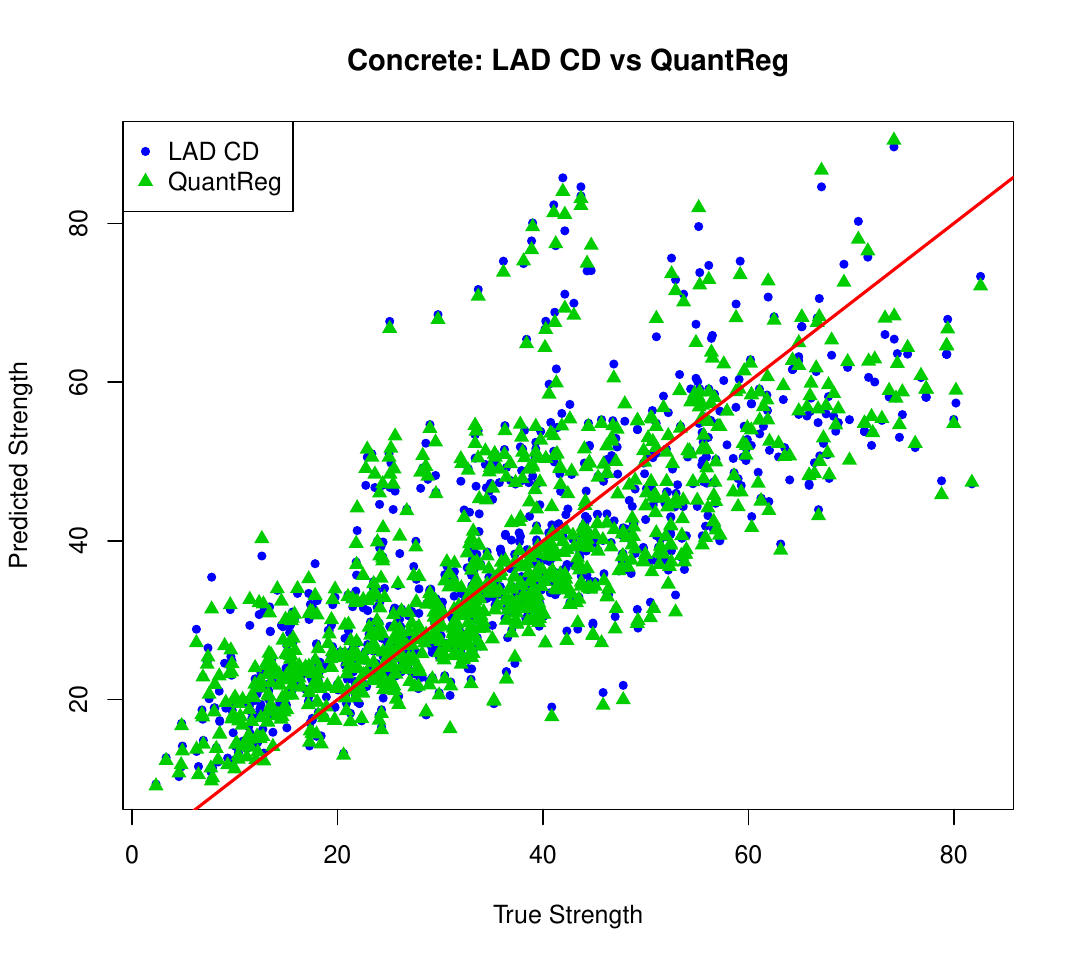}
    \hfill
    \includegraphics[width=0.45\linewidth]{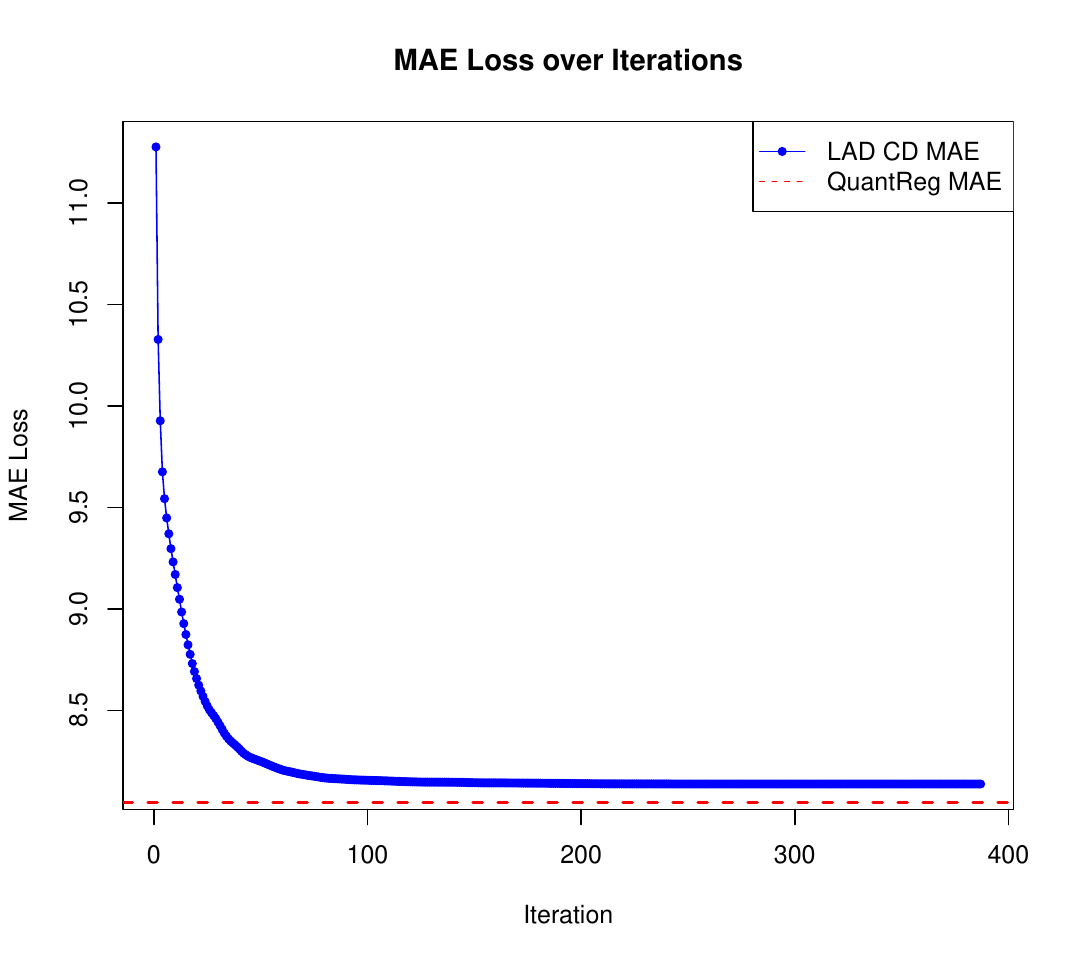}
    \caption{[Left] Predicted vs.\ true compressive strengths for the Concrete dataset. Both LAD-CD (blue) and QuantReg (green) yield similar predictive accuracy. [Right] MAE loss over iterations for LAD-CD compared to QuantReg. The coordinate descent method converges monotonically toward the quantile regression baseline.}
    \label{fig:concrete_pred}
\end{figure}

\section{Conclusion}
\label{sec:conclusion}

We proposed a simple and scalable coordinate descent algorithm for Least Absolute Deviations (LAD) regression. The method exploits the structure of the LAD objective: each coordinate update reduces to an exact one-dimensional minimization solved by a median or weighted median. Combined with incremental residual updates, this yields a numerically lightweight algorithm that is provably descent-based and easy to implement.

Our results highlight several key findings. First, LAD-CD provides robust estimation under heavy-tailed noise and outlier contamination, matching the predictive accuracy of LP-based quantile regression solvers when applicable and substantially outperforming ordinary least squares in contaminated settings. Second, unlike LP-based approaches, LAD-CD remains well-defined and stable in high-dimensional and underdetermined regimes ($p \ge n$), as it avoids matrix inversion entirely. Third, the optimized implementation reduces the cost of a full coordinate sweep from $O(np^2)$ to $O(p n \log n)$, yielding substantial practical speedups for large $p$ while retaining monotone convergence.

Initialization plays a central role in finite-sample performance. Across all experiments, ridge regression provided the most effective warm start, consistently delivering the best trade-off between runtime and predictive accuracy. Genetic-algorithm initializations were occasionally useful in small or highly nonconvex settings. Still, they scaled poorly with dimension, while random multi-start strategies offered a simple fallback when closed-form initializations were unavailable.

\paragraph{Practical guidance.}
For problems with $p < n$, LP-based quantile regression remains a fast and reliable baseline. In higher-dimensional settings, we recommend ridge-initialized LAD-CD as a default choice: it is stable, robust, and computationally efficient. Random multi-starts can be used when ridge is unavailable, while GA-based warm starts should be reserved for small-scale problems where global exploration is essential.

\paragraph{Limitations and future work.}
Several extensions merit further study. Computationally, linear-time or approximate median updates, parallel coordinate schemes, and lower-level implementations could further reduce runtime. Methodologically, incorporating explicit regularization (e.g., $\ell_1$ or $\ell_2$ penalties) within the coordinate updates would broaden applicability. On the theoretical side, finite-sample guarantees for ridge-initialized LAD-CD in high dimensions and efficiency comparisons with quantile regression under misspecified noise models remain open problems.

We emphasize that neither the formulation of the algorithm nor the convergence analysis relies on distributional assumptions for the error terms beyond those implicit in the LAD objective. In particular, no finite-variance condition is required for the algorithmic guarantees. Future work may investigate the statistical performance of the proposed method under minimal assumptions on the errors, such as independence and zero conditional median, without imposing moment conditions. Establishing finite-sample guarantees under such nonparametric settings remains an interesting direction for further study.

\paragraph{Final remark.}
LAD-CD offers a practical alternative for robust linear regression: it is modular, solver-free, scalable to high dimensions, and, when paired with a simple ridge warm start, delivers consistently strong empirical performance. We expect it to be useful both as a teaching tool and as a robust baseline for applied regression problems.

\paragraph{Data Availability Statement:}
The data and code that support the findings of this study are available at the following \href{https://github.com/Zehaan22/Coordinate_Descent_for_LAD/}{GitHub repository}.

\bibliography{references}

@article{charnes1955,
  title={Optimal estimation of executive compensation by linear programming},
  author={Charnes, A. and Cooper, W. W. and Ferguson, R. O.},
  journal={Management Science},
  volume={1},
  number={2},
  pages={138--151},
  year={1955},
  publisher={INFORMS}
}

@article{koenker1978,
  title={Regression quantiles},
  author={Koenker, Roger and Bassett, Gilbert, Jr},
  journal={Econometrica: Journal of the Econometric Society},
  pages={33--50},
  year={1978},
  publisher={JSTOR}
}

@article{BarrodaleRoberts1973,
  author  = {Iain Barrodale and F. D. K. Roberts},
  title   = {An Improved Algorithm for Discrete $L_{1}$ Linear Approximation},
  journal = {SIAM Journal on Numerical Analysis},
  year    = {1973},
  volume  = {10},
  number  = {5},
  pages   = {839--848},
  doi     = {10.1137/0710069}
}

@article{KoenkerDOrey1987,
  author  = {Roger W. Koenker and Vasco d'Orey},
  title   = {Computing Regression Quantiles},
  journal = {Journal of the Royal Statistical Society. Series C (Applied Statistics)},
  year    = {1987},
  volume  = {36},
  number  = {3},
  pages   = {383--393},
  doi     = {10.2307/2347802}
}

@book{Holland1975,
  author    = {John H. Holland},
  title     = {Adaptation in Natural and Artificial Systems},
  year      = {1975},
  publisher = {University of Michigan Press},
  address   = {Ann Arbor, MI}
}

@book{Goldberg1989,
  author    = {David E. Goldberg},
  title     = {Genetic Algorithms in Search, Optimization and Machine Learning},
  year      = {1989},
  publisher = {Addison-Wesley},
  address   = {Reading, MA}
}

@phdthesis{DeJong1975,
  author       = {K. A. De Jong},
  title        = {An Analysis of the Behavior of a Class of Genetic Adaptive Systems},
  year         = {1975},
  school       = {University of Michigan},
  note         = {Ph.D. thesis}
}

@book{Vose1999,
  author    = {Michael D. Vose},
  title     = {The Simple Genetic Algorithm: Foundations and Theory},
  year      = {1999},
  publisher = {MIT Press},
  address   = {Cambridge, MA}
}

@article{Friedman2010,
  author = {Friedman, J. and Hastie, T. and Tibshirani, R.},
  title = {Regularization Paths for Generalized Linear Models via Coordinate Descent},
  journal = {Journal of Statistical Software},
  year = {2010},
  volume = {33},
  number = {1},
  pages = {1--22}
}

@article{HoerlKennard1970,
  title = {Ridge Regression: Biased Estimation for Nonorthogonal Problems},
  author = {Hoerl, A. E. and Kennard, R. W.},
  journal = {Technometrics},
  year = {1970},
  volume = {12},
  number = {1},
  pages = {55--67},
  doi = {10.1080/00401706.1970.10488634}
}

@inproceedings{JamesStein1961,
  title = {Estimation with Quadratic Loss},
  author = {James, W. and Stein, C.},
  booktitle = {Proceedings of the Fourth Berkeley Symposium on Mathematical Statistics and Probability, Volume 1: Contributions to the Theory of Statistics},
  year = {1961},
  pages = {361--379},
  publisher = {University of California Press}
}

@article{EfronMorris1977,
  title = {Stein's Estimation Rule and Its Competitors—An Empirical Bayes Approach},
  author = {Efron, B. and Morris, C.},
  journal = {Journal of the American Statistical Association},
  year = {1977},
  volume = {68},
  number = {341},
  pages = {117--130},
  doi = {10.1080/01621459.1973.10482467}
}

@article{Tseng2001,
  author    = {Paul Tseng},
  title     = {Convergence of a Block Coordinate Descent Method for Nondifferentiable Minimization},
  journal   = {Journal of Optimization Theory and Applications},
  year      = {2001},
  volume    = {109},
  number    = {3},
  pages     = {475--494},
  doi       = {10.1023/A:1017501703105}
}

@book{hastie2009,
  title     = {The Elements of Statistical Learning: Data Mining, Inference, and Prediction},
  author    = {Hastie, Trevor and Tibshirani, Robert and Friedman, Jerome},
  year      = {2009},
  edition   = {2},
  publisher = {Springer},
  address   = {New York}
}

@book{koenker2005,
  title     = {Quantile Regression},
  author    = {Koenker, Roger},
  year      = {2005},
  publisher = {Cambridge University Press},
  address   = {Cambridge}
}

@article{portnoy1997,
  title   = {Asymptotic behavior of regression quantiles in high dimensions},
  author  = {Portnoy, Stephen},
  journal = {Journal of Multivariate Analysis},
  volume  = {62},
  number  = {2},
  pages   = {239--254},
  year    = {1997}
}

@book{huber2009,
  title     = {Robust Statistics},
  author    = {Huber, Peter J. and Ronchetti, Elvezio M.},
  publisher = {Wiley},
  year      = {2009}
}

@article{yeh1998,
  title   = {Modeling of strength of high-performance concrete using artificial neural networks},
  author  = {Yeh, I.-Cheng},
  journal = {Cement and Concrete Research},
  volume  = {28},
  number  = {12},
  pages   = {1797--1808},
  year    = {1998}
}
\bibliographystyle{plainnat}

\section*{Appendix}
\subsection*{Proofs for descent property}
\label{app:convergence}

\begin{proof}[Proof of Lemma~\ref{lem: Convex LAD}]
The function $t \mapsto |t|$ is convex, and each residual $r_i(\beta) = y_i - x_i^\top \beta$ is affine in $\beta$. The composition of a convex function with an affine map is convex, and a nonnegative sum of convex functions is convex. Hence $L$ is convex.
\end{proof}

\begin{proof}[Proof of Theorem~\ref{THM: LAD_CD-Convergence}]
Let $\beta^{k,j}$ denote the parameter vector after updating coordinate $j$ in iteration $k$. Since each update solves $\beta_j^{k+1} = \arg\min_{t} g_j(t)$ exactly, we have
\[
\mathcal{L}(\beta^{k,j}) \leq \mathcal{L}(\beta^{k,j-1}), \qquad j = 0, \dots, p.
\]
Chaining over all $p$ coordinate updates within an iteration yields
\[
\mathcal{L}(\beta^{k+1}) = \mathcal{L}(\beta^{k,p}) \leq \dots \leq \mathcal{L}(\beta^{k,0}) = \mathcal{L}(\beta^{k}).
\]
Thus, the sequence $\{\mathcal{L}(\beta^{k})\}$ is monotonically non-increasing and bounded below by zero, and therefore convergent. Since $L$ is convex and each coordinate update yields an exact minimizer of the corresponding subproblem, any limit point $\beta^{\star}$ must be a global minimizer of $L$ \citep[see][]{Tseng2001}.
\end{proof}

\subsection*{Additional results}
\begin{table}[!ht]
\centering
\caption{Final prediction MAE of LAD--CD for $p = 100$ under varying sample sizes. 
Results are reported for warm-start initialization (true value + gaussian perturbation) and uninformed (zero) initialization.}
\label{tab:initial_mae_scaling}
\begin{tabular}{ccc}
\toprule
\textbf{$n$} & \textbf{Warm Start MAE} & \textbf{Uninformed Start MAE} \\
\midrule
10    & 0.0893 & 1.8804 \\
20    & 1.6874 & 0.6852 \\
50    & 1.4514 & 3.3592 \\
100   & 1.4361 & 5.1271 \\
200   & 0.7401 & 1.2320 \\
500   & 0.6986 & 0.7000 \\
1000  & 0.7002 & 0.7011 \\
10000 & 0.7982 & 0.7979 \\
\bottomrule
\end{tabular}
\end{table}

\end{document}